\pdfoutput=1 
\documentclass[11pt,draftcls,onecolumn]{IEEEtran}
\usepackage{amssymb, amsmath, graphicx, cite, color, epsfig, subfigure,wrapfig}

\newcommand{\bfX}{ {\bf X} }

\newtheorem{theorem}{Theorem}
\newtheorem{lemma}[theorem]{Lemma}

\newtheorem{remark}{Remark}

\allowdisplaybreaks[1]

\begin{document}
\title{Lattice Coding for the Two-way Two-relay Channel}
\author{Yiwei Song, Natasha Devroye, Huai-Rong Shao and Chiu Ngo
\thanks{Yiwei Song and
Natasha Devroye are with the Department of Electrical and Computer
Engineering, University of Illinois at Chicago, Chicago, IL 60607.
Email: ysong34, devroye@uic.edu. Huai-Rong Shao and Chiu Ngo are with Samsung Electronics, US R\&D Center (SISA)
San Jose, CA 95134.
Email: {hr.shao, chiu.ngo}@samsung.com. Portions of this work were performed at Samsung Electronics, US R\&D Center during Yiwei Song's summer internship in 2012.}
}
\maketitle

\abstract 
Lattice coding techniques may be used to derive achievable rate regions which outperform known independent, identically distributed (i.i.d.) random codes in multi-source relay networks and in particular the two-way relay channel. Gains stem from the ability to decode the sum of codewords (or messages) using lattice codes at higher rates than possible with i.i.d. random codes. Here we develop a novel  lattice coding scheme for the Two-way Two-relay Channel: $1 \leftrightarrow 2\leftrightarrow 3 \leftrightarrow 4$, where Node $1$ and $4$ simultaneously communicate with each other through two relay nodes $2$ and $3$.  Each node only communicates with its neighboring nodes. 
The key technical contribution is the lattice-based achievability strategy, where each relay is able to remove the noise while decoding the sum of several signals in a Block Markov strategy and then re-encode the signal into another lattice codeword using the so-called ``Re-distribution Transform''. This allows nodes further down the line to again decode sums of lattice codewords. 
This transform is central to improving the achievable rates, and ensures that the messages traveling in each of the two directions fully utilize the relay's power, even under asymmetric channel conditions. All decoders are lattice decoders and only a single nested lattice codebook pair is needed. 
The symmetric rate  achieved by the proposed lattice coding scheme is within $\frac{1}{2} \log 3$ bit/Hz/s of the symmetric rate capacity. 

\section{Introduction}
\label{sec:intro}

Lattice codes may be viewed as linear codes in Euclidean space;  the sum of two lattice codewords is still a codeword. This group property is exploited in several additive white Gaussian noise (AWGN) relay networks, such as \cite{Song:lattice, Nazer2011compute, Nam:IEEE, Narayanan:2010}, where it has been shown that lattice codes may outperform i.i.d. random codes in certain scenarios, particularly when interested in decoding a linear combination of the received codewords rather than the individual codewords. 
One such example  is the AWGN Two-way Relay channel, where two users wish to communicate with each other through a relay node \cite{Nam:IEEE, Narayanan:2010}. If the two users employ properly chosen lattice codewords, the relay node may decode the sum of the codewords from both users directly at higher rates than decoding them individually. 
It is then sufficient for the relay to broadcast this sum of codewords to both users since each user, knowing the sum and its own message, may determine the other desired message. 

{\bf Past work.}  Lattice codes have been used to derive achievable rates for multi-user networks beyond the two-way relay channel. Nested lattice codes have been shown to be capacity achieving in the point-to-point Gaussian channel \cite{Erez:2004}, the Gaussian Multiple-access Channel \cite{Nazer2011compute}, Broadcast Channel \cite{Zamir:2002:binning}, and to achieve the same rates as those achieved by i.i.d. Gaussian codes in the Decode-and-Forward rate  \cite{Song:lattice, Nokleby2012} and Compress-and-Forward rates  \cite{Song:lattice} of the Relay Channel \cite{Cover:1979}. Lattice codes may further be used in achieving the capacity of Gaussian channels with interference or state known at the transmitter (but not receiver) \cite{gelfand} using a lattice equivalent \cite{Zamir:2002:binning} of dirty-paper coding (DPC) \cite{costa}.  The nested lattice approach of  \cite{Zamir:2002:binning} for the dirty-paper channel is extended to dirty-paper networks in \cite{ Philosof:DirtyMAC}.
 Lattice codes have also been shown to approximately achieve the capacity of the K-user interference channel \cite{ordentlich2012}, using the Compute-and-Forward framework for decoding linear equations of codewords (or equivalently messages) of \cite{Nazer2011compute}. 

We consider the Two-way Two-relay Channel: $1 \leftrightarrow 2\leftrightarrow 3 \leftrightarrow 4$ where two user Nodes $1$ and $4$ exchange information with each other through the relay nodes $2$ and $3$. 
This is related to the work of \cite{Popovski:2006a}, which considers the throughput of  i.i.d. random code-based Amplify-and-Forward and Decode-and-Forward approaches for this channel model, or the i.i.d. random coding based schemes of  \cite{Kim:ISIT2009} and \cite{Kim:multi} where there are furthermore links between all nodes. 
This model is also different from that in \cite{Pooniah:2008} where a two-way relay channel with two {\it parallel} (rather than sequential as in this work) relays are considered. 

{\bf Contributions.} The Two-way Two-relay Channel $1 \leftrightarrow 2\leftrightarrow 3 \leftrightarrow 4$ is a generalization of the Two-way Relay channel $1 \leftrightarrow 2\leftrightarrow 3$   to multiple relays. As lattice codes proved useful in the Two-way Relay channel in decoding the sum at the relay, it is natural to expect lattice codes to again perform well for the considered channel. 
We propose, for the first time, a lattice based scheme for this two-way line network, where all nodes transmit lattice codewords and each relay node decodes a sum of these codewords. 
This scheme may be seen as a generalization of the lattice based scheme of \cite{Nam:IEEE, Narayanan:2010} for the Two-way Relay Channel. However, this generalization is not straightforward due to the presence of multiple relays and hence the need to repeatedly be able to decode the sum of codewords. One way to enable this is to have the relays employ lattice codewords as well -- something not required in the Two-way Relay channel. 
In the Two-way Relay channel achievability schemes consists of two phases -- the multiple access and the broadcast phase. The scheme includes multiple Block Markov phases, where 
during each phase, the end users send new messages encoded by lattice codewords and the relays decode a combination of lattice codewords. The relays then perform a ``Re-distribution Transform" on the decoded lattice codeword combinations, and broadcasts the resulting lattice codewords. 
The novelty of our scheme lies in this ``Re-distribution Transform" which enables the messages traveling in both directions to fully utilize the relay transmit power. Furthermore, all decoders are lattice decoders (more computationally efficient than joint typicality decoders) and only a single nested lattice codebook pair is needed. 

{\bf Outline.} We first outline some definitions, notation, and technical lemmas for nested lattice codes in Section \ref{sec:prelim}. We outline the channel model in Section \ref{sec:model}. We then outline a lattice based strategy for the broadcast phase of the Two-way (single) Relay channel with asymmetric uplinks in Section \ref{sec:BC}, which includes the key technical novelty -- the ``Re-distribution Transform'' which allows relays to intuitively spread the signals traveling in both directions to utilize the relay's entire transmit power. We present the main achievable rate regions for the Two-way Two-relay channel in Section \ref{sec:procedure} before outlining extensions to half-duplex nodes and more than two relays in Section \ref{sec:extra}. Finally, we conclude in Section \ref{sec:conclusion}.


\section{Preliminaries on lattice codes and notation}
\label{sec:prelim}

We now define our notation for lattice codes in Subsection \ref{subsec:lattice}, define properties of nested lattice codes in Subsection \ref{subsec:nested} and technical lemmas in Subsection \ref{subsec:lemma}. 

\subsection{Lattice codes}
\label{subsec:lattice}


Our notation for (nested) lattice codes for transmission over AWGN channels  follows that of  \cite{Zamir:2002:binning, nam:2009nested}; comprehensive treatments may be found in \cite{loeliger1997averaging, Zamir:2002:binning,  Erez:2004} and in particular \cite{zamir-lattices}.  An $n$-dimensional lattice $\Lambda$ is a discrete subgroup of Euclidean space $\mathbb{R}^n$ with Euclidean norm $|| \cdot ||$ under vector addition and may be expressed as all integral combinations of basis vectors ${\bf b_i}\in {\mathbb R}^n$
\[ \Lambda = \{ \lambda = {\bf B} \; {\bf i}: \; {\bf i}\in \mathbb{Z}^n\},\]
for $\mathbb{Z}$ the set of integers, $n\in {\mathbb Z}_+$, and ${\bf B} := [{\bf b_1} | {\bf b_2}| \cdots {\bf b_n}]$ the $n\times n$ generator matrix corresponding to the lattice $\Lambda$.  We use bold ${\bf x}$ to denote  column vectors, ${\bf x}^T$ to denote the transpose of the vector ${\bf x}$. All vectors lie in ${\mathbb R}^n$ unless otherwise stated, and all logarithms are base 2. Let ${\bf 0}$ denote the all zeros vector of length $n$, ${\bf I}$ denote the $n\times n$ identity matrix, and ${\cal N}(\mu,\sigma^2)$ denote a Gaussian random variable (or vector) of mean $\mu$ and variance $\sigma^2$.  Let $|{\cal C}|$ denote the cardinality of the set ${\cal C}$.  Define $C(x): = \frac{1}{2}\log_2\left(1+x \right)$. Further define or note that

$\bullet$ The {\it nearest neighbor lattice quantizer} of $\Lambda$ as \[ Q({\bf x}) = \arg \min_{\lambda\in \Lambda} ||{\bf x}-\lambda||;\]

$\bullet$ The $\mod \Lambda$ operation as ${\bf x} \mod \Lambda : = {\bf x} - Q({\bf x})$;

$\bullet$ The {\it Voronoi region of $\Lambda$} as the points closer to the origin than to any other lattice point \[\mathcal{V}:= \{{\bf x}:Q({\bf x}) = {\bf 0}\},\]
which is of volume $V: = \mbox{Vol}({\mathcal V})$ (also sometimes denoted by $V(\Lambda)$ or $V_i$ for lattice $\Lambda_i$);

$\bullet$ The {\it second moment per dimension of a uniform distribution over ${\mathcal V}$} as
\[ \sigma^2(\Lambda) : = \frac{1}{V}\cdot \frac{1}{n} \int_{\mathcal V} ||{\bf x}||^2 \; d{\bf x};\]

$\bullet$ For any ${\bf s} \in \mathbb{R}^n$ ,
\begin{align}
( \alpha ( {\bf  s} \mod \Lambda )) \mod \Lambda &= ( \alpha {\bf s} ) \mod \Lambda, \quad  \alpha \in \mathbb{Z}.  \label{eq:operation1}\\
\beta ( {\bf s} \mod \Lambda ) &=  ( \beta {\bf s}) \mod \beta \Lambda,  \quad \beta \in \mathbb{R}. \label{eq:operation2}
\end{align}

$\bullet$ The definitions of Rogers good and Poltyrev good lattices are stated in \cite{Song:lattice}; we will not need these definitions explicitly. Rather, we will use the results derived from lattices with these properties.

\subsection{Nested lattice codes}
\label{subsec:nested}

Consider two lattices $\Lambda$ and $\Lambda_c$ such that $\Lambda \subseteq \Lambda_c$ with fundamental regions ${\cal V}, {\cal V}_c$ of volumes $V, V_c$ respectively. Here $\Lambda$ is termed the {\it coarse} lattice which is a sublattice of  $\Lambda_c$,  the {\it fine} lattice, and hence $V \geq V_c$.  When transmitting over the AWGN channel, one may use the set  $ \mathcal{C}_{\Lambda_c, {\cal V}} = \{ \Lambda_c \cap \mathcal{V} (\Lambda) \} $ as the codebook. The coding rate $R$ of this {\it nested  $(\Lambda, \Lambda_c)$ {lattice pair}} is defined as
\[ R = \frac{1}{n} \log |\mathcal{C}_{\Lambda_c, {\cal V}}| = \frac{1}{n} \log \frac{V}{V_c}.\] 
{Nested lattice pairs} were shown to be capacity achieving for the AWGN channel \cite{Erez:2004}.

In this work, we only need one ``good'' nested lattice pair $\Lambda \subseteq \Lambda_c$, in which $\Lambda$ is both Rogers good and Poltyrev good and $\Lambda_c$ is Poltyrev good (see definitions of these in \cite{Song:lattice}). This lattice pair is used throughout the paper. 
The existence of such a pair of nested lattices may be guaranteed by \cite{Erez:2004}; we now describe the construction procedure of such a good nested lattice pair. Suppose the coarse lattice $\Lambda = {\bf B} Z^n$ is both Rogers good and Poltyrev good with second moment $\sigma^2(\Lambda) =1$.  With respect to this coarse lattice $\Lambda$, the fine lattice $\Lambda_c$ is generated by Construction A \cite{Erez:2004, Nazer2011compute}, which maps a codebook of a linear block code over a finite field into real lattice points. 
The generation procedure is:
\begin{itemize}
\item Consider the vector $ {\bf G} \in \mathbb{F}_{P_{prime}}^n$ with every element drawn from an i.i.d. uniform distribution from the finite field of order $P_{prime}$ (a prime number) which we take to be $\mathbb{F}_{P_{prime}} = \{ 0,1,2, \dots, P_{prime}-1\} $ under addition and multiplication modulo $P_{prime}$.
\item The codebook $\bar{\mathcal{C}}$ of the linear block code induced by {\bf G} is $\bar{\mathcal{C}} = \{ \bar{c} = {\bf G}w: w \in \mathbb{F}_{P_{prime}} \}$. 
\item Embed this codebook into the unit cube by scaling down by a factor of $P_{prime}$ and then place a copy at every integer vector: $\bar{\Lambda}_c = P_{prime}^{-1} \bar{\mathcal{C}} + \mathbb{Z}^n$.
\item Rotate $\bar{\Lambda}_c$ by the generator matrix of the coarse lattice to obtain the desired fine lattice: $\Lambda_c= {\bf B} \bar{\Lambda}_c$.
\end{itemize}
Now let $\phi(\cdot)$ denote the one-to-one mapping between one element in the one dimensional finite field $w \in \mathbb{F}_{P_{prime}}$ to a point in $n$-dimension real space ${\bf t} \in  \mathcal{C}_{\Lambda_c, {\cal V}}$:
\begin{equation}
{\bf t}= \phi(w) = ( {\bf B} P_{prime}^{-1} {\bf G}w )\mod \Lambda,
\label{eq:phi}
\end{equation}
 with inverse mapping $w = \phi^{-1}({\bf t}) = ( {\bf G}^T {\bf G} )^{-1} {\bf G}^T ( P_{pirme} ( {\bf B}^{-1} \bf t \mod \mathbb{Z}^n) )$ (see \cite[Lemma 5 and 6]{Nazer2011compute}). The mapping operation $\phi(\cdot)$ defined here is used in the following lemmas.  


\subsection{Technical lemmas}
\label{subsec:lemma}

We first state several lemmas needed in the proposed two-way lattice based scheme. In the following,  ${\bf t}_{ai}$ and ${\bf t}_{bi} \in  \mathcal{C}_{\Lambda_c, {\cal V}}$ are generated from $w_{ai}$ and $w_{bi} \in \mathbb{F}_{P_{prime}}$  as ${\bf t}_{ai}= \phi(w_{ai}), {\bf t}_{bi}= \phi(w_{bi})$. Furthermore, let $\alpha, \alpha_i, \beta_i \in \mathbb{Z}$ such that $\frac{\alpha}{P_{prime}}, \frac{\alpha_i}{P_{prime}}, \frac{\beta_i}{P_{prime}} \notin \mathbb{Z}$ and $\theta \in \mathbb{R}$. We use $\oplus$, $\otimes$ and $\ominus$ to denote modulo $P_{prime}$ addition, multiplication, and subtraction over the finite field $\mathbb{F}_{P_{prime}}$.

\begin{lemma}
\label{lem:mapping}
There exists an one-to-one mapping between ${\bf v} = ( \sum_i \alpha_i \theta {\bf t}_{ai} + \sum_i \beta_i \theta {\bf t}_{bi} ) \mod \theta \Lambda$ and ${\bf u} = \bigoplus_i \alpha_i w_{ai} \oplus \bigoplus_i \beta_i w_{bi}$. 
\end{lemma}
\begin{proof}
The proof follows from \cite[Lemma 6]{Nazer2011compute}, where it is shown that there is an one-to-one mapping between $ \theta^{-1} {\bf v} = ( \sum_i \alpha_i {\bf t}_{ai} + \sum_i \beta_i  {\bf t}_{bi} ) \mod  \Lambda$ and ${\bf u}' = \bigoplus_i \alpha'_i w_{ai} \oplus \bigoplus_i \beta'_i w_{bi}$ 
, where $\alpha'_i = \alpha_i \mod P_{prime}$ and $\beta'_i = \beta_i \mod P_{prime}$. This one-to-one mapping is given by $ \theta^{-1} {\bf v} = \phi^{-1}({\bf u}')$ and ${\bf u}' = \phi(\theta^{-1}{\bf v})$. 
Then  observe that $\bigoplus_i \alpha_i w_{ai} \oplus \bigoplus_i \beta_i w_{bi} = \bigoplus_i \alpha'_i w_{ai} \oplus \bigoplus_i \beta'_i w_{bi}$ by the properties of modulo addition and multiplication \cite{}. 
Thus, the one-to-one mapping between ${\bf v}$ and ${\bf u}$ is given by $ {\bf v} =\theta  \phi^{-1}({\bf u})$ and ${\bf u} = \phi(\theta^{-1}{\bf v})$. 
\end{proof}

\begin{lemma}
\label{lem:prime}
There exists an one-to-one mapping between $\alpha \otimes w$ and $w$. 
\end{lemma}

\begin{proof}
Recall that $\frac{\alpha}{P_{prime}} \notin \mathbb{Z}$. 
Suppose $\alpha \otimes w_1 = \alpha \otimes w_2$. Then $\alpha(w_1 - w_2) = \kappa P_{prime}$ for some integer $\kappa$. Re-writing this, we have $\frac{\alpha}{P_{prime}} (w_1-w_2) = \kappa$ for some integer $\kappa$. This implies that either $\frac{\alpha}{P_{prime}} \in \mathbb{Z}$ or $\frac{w_1 - w_2}{P_{prime}} \in \mathbb{Z}$. Since $\frac{\alpha}{P_{prime}} \notin \mathbb{Z}$, $\frac{w_1 - w_2}{P_{prime}} \in \mathbb{Z}$. Thus $w_1 - w_2 =0$ since $-(P_{prime}-1) \leq w_1 - w_2 \leq P_{prime} -1$. Hence $w_1 =w_2$.
\end{proof}

\medskip
\begin{lemma}
\label{lem:codingrate}
If $w_{ai}$ and $w_{bi}$ are uniformly distributed over $\mathbb{F}_{P_{prime}}$,  then $( \sum_i \alpha_i \theta {\bf t}_{ai} + \sum_i \beta_i \theta {\bf t}_{bi} ) \mod \theta \Lambda$ is uniformly distributed over $\{ \theta \Lambda_c \cap  \mathcal{V} (\theta \Lambda) \}$.
\end{lemma}
\begin{proof}
As in the proof of Lemma \ref{lem:mapping}, $(\sum_i \alpha_i \theta {\bf t}_{ai} + \sum_i \beta_i \theta {\bf t}_{bi} ) \mod \theta \Lambda  = \theta (  (\sum_i \alpha_i  {\bf t}_{ai} + \sum_i \beta_i {\bf t}_{bi} ) \mod \Lambda ) = \theta \phi (\bigoplus_i \alpha_i w_{ai} \oplus \bigoplus_i \beta_i w_{bi})$. Since $w_{ai}$ and $w_{bi}$ are uniformly distributed over $\mathbb{F}_{P_{prime}}$, $ \alpha_i w_{ai}$ and $\beta_i w_{bi}$ are uniformly distributed over $\mathbb{F}_{P_{prime}}$ by Lemma \ref{lem:prime}.  Then, $\bigoplus_i \alpha_i w_{ai} \oplus \bigoplus_i \beta_i w_{bi}$ is uniformly distributed over $\mathbb{F}_{P_{prime}}$, and  $ \phi (\bigoplus_i \alpha_i w_{ai} \oplus \bigoplus_i \beta_i w_{bi})$ is uniformly distributed over $\{ \Lambda_c \cap \mathcal{V} (\Lambda) \}$, and finally $\theta \phi (\bigoplus_i \alpha_i w_{ai} \oplus \bigoplus_i \beta_i w_{bi})$ is uniformly distributed over $\{ \theta \Lambda_c \cap \mathcal{V} (\theta \Lambda) \}$.
\end{proof}


\section{Channel Model}
\label{sec:model}

\begin{figure}
\centering
\includegraphics[width=12cm]{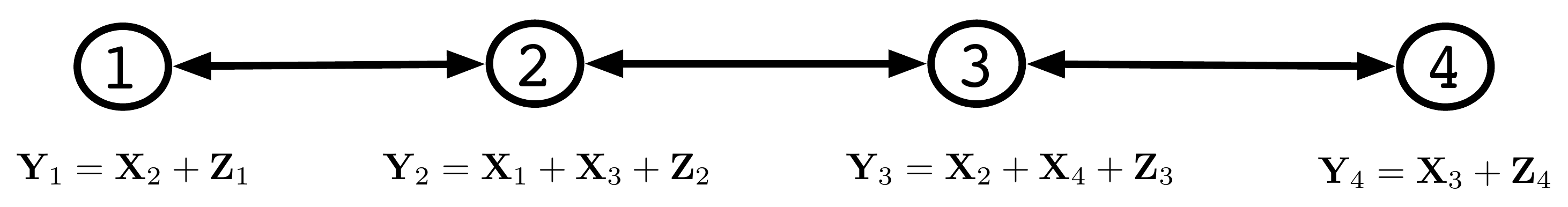}
\caption{The Gaussian  Two-way Two-relay Channel Model.}
\label{fig:model}
\end{figure}

The Gaussian  Two-way Two-relay Channel describes a wireless communication scenario where two source nodes (Node $1$ and $4$) 
simultaneously communicate with each other through multiple full-duplex relays (Node $2$ and $3$) and multiple hops as shown in Figure \ref{fig:model}. Each node can only communicate with its neighboring nodes. The channel model may be expressed as (all bold symbols are $n$ dimensional)
\begin{align*}
{\bf Y}_1 &= {\bf X}_2 + {\bf Z}_1 \\
{\bf Y}_2 &= {\bf X}_1 + {\bf X}_3 + {\bf Z}_2\\
{\bf Y}_3 &= {\bf X}_2 + {\bf X}_4 + {\bf Z}_3\\
{\bf Y}_4 &= {\bf X}_3 + {\bf Z}_4
\end{align*}
where ${\bf Z_i} $ ($i \in \{1,2,3,4\}$) is an i.i.d. Gaussian noise vector with variance $N_i$: ${\bf Z}_i \sim \mathcal{N} ( {\bf 0}, N_i {\bf  I})$, and the input ${\bf X}_i$ is subject to the transmit power constraint $P_i$: $\frac{1}{n} E({\bf X}_i^T{\bf X}_i) \leq P_i$. Note that since we can always subtract the signal transmitted by the node itself, they are omitted in the channel model expression. Also note the arbitrary power constraints and noise variances but unit channel gains. 

An $(2^{nR_a}, 2^{nR_b}, n)$ code for the Gaussian Two-way Two-relay channel consists of the two sets of messages $w_a, w_b$ uniformly distributed 
over  ${\cal M}_a: = \{1,2,\cdots , 2^{nR_a}\}$ and ${\cal M}_b: = \{1,2,\cdots , 2^{nR_b}\}$ respectively,  and two encoding functions $X_1^n: {\cal M}_a \rightarrow {\mathbb R}^n$ (shortened to ${\bf X_1}$) and $X_4^n: {\cal M}_b \rightarrow {\mathbb R}^n$ (shortened to ${\bf X_4}$), satisfying the 
power constraints $P_1$ and $P_4$ respectively, two sets of relay functions $\{f_{k,j}\}_{j=1}^n$ ($k=2,3$) such that the relay channel input at time $j$ is a function of the 
previously received relay channel outputs from channel uses $1$ to $j-1$, $X_{k,j} = f_{k,j}(Y_{k,1}, \cdots,Y_{k,j-1})$, and finally two 
decoding functions $g_1: {\cal Y}_1^n \times {\cal M}_a \rightarrow {\cal M}_{b}$ and  $g_4: {\cal Y}_4^n \times {\cal M}_b \rightarrow {\cal M}_{a}$ which yield the message estimates $\hat{w}_{b}: = g_1(Y_1^n, w_a)$ and $\hat{w}_{a}: = g_2(Y_4^n, w_b)$ respectively. 
We define the average probability of error of the code to be $P_{n,e} : = \frac{1}{2^{n(R_a+R_b)}} \sum_{w_a\in {\cal M}_a, w_b\in {\cal M}_b} \Pr\{(\hat{w_a}, \hat{w_b}) \neq (w_a,w_b)|(w_a,w_b) \mbox{ sent}\}$. The rate 
pair $(R_a,R_b)$ is then said to be achievable by the two-relay channel if, for any $\epsilon>0$ and for sufficiently large $n$, there 
exists an $(2^{nR_a},2^{nR_b},n)$ code such that $P_{n,e} < \epsilon$. The capacity region of the Gaussian Two-way Two-relay channel is the 
supremum of the set of achievable rate pairs.

\section{Lattice Codes in the BC Phase of the Two-way Relay Channel}
\label{sec:BC}

The work \cite{Nam:IEEE, Narayanan:2010} introduces a two-phase lattice scheme for the Gaussian Two-way Relay Channel $1 \leftrightarrow 2 \leftrightarrow 3$, where two user nodes $1$ and $3$ exchange information through a single relay node $2$ (all definitions are analogous to the previous section):  
 the Multiple-access Channel (MAC) phase and the Broadcast Channel (BC) phase. In the MAC phase, the relay receives a noisy version of the sum of two signals from both users as in a multiple access channel (MAC). If the codewords are from nested lattice codebooks, the relay may decode the sum of the two codewords directly without decoding them individually. This is sufficient for this channel, as then, in the BC phase, the relay may broadcast the sum of the codewords to both users who may determine the other message using knowledge of  their own transmitted message.  In the scheme of \cite{Nam:IEEE}, the relay re-encodes the decoded sum of the  codewords into a codeword from an i.i.d. random codebook while \cite{Narayanan:2010} uses a lattice codebook in the downlink. 

In extending the schemes of \cite{Nam:IEEE, Narayanan:2010} to multiple relays we would want to use lattice codebooks in the BC phase, as in \cite{Narayanan:2010}. This would,  for example, allow the signal sent by Node $2$  to be aligned  with Node $4$'s transmitted signal (aligned is used to mean that the two codebooks are nested) in the Two-way Two-relay Channel: $1 \leftrightarrow 2 \leftrightarrow 3  \leftrightarrow 4$ and hence enable the decoding of the sum of codewords again at Node 3.  However, the scheme of \cite{Narayanan:2010} is only applicable to channels in which the SNR from the users to the relay are symmetric, i.e. $\frac{P_1}{N_2} = \frac{P_3}{N_2}$ . In this case  the relay can simply broadcast the decoded (and possibly scaled) sum of codewords sum without re-encoding it.  Thus, before tackling the Two-way Two-relay channel, we first devise a lattice-coding scheme for the  BC phase in the Two-way Relay Channel with {\it arbitrary}  uplink SNRs $\frac{P_1}{N_2} \neq \frac{P_3}{N_2}$ . 


\subsection{Lattice codes for the BC phase for the general Gaussian Two-way Relay Channel}


In this section, we design a lattice coding scheme for the BC phase of the Gaussian Two-way (single-relay) Channel with arbitrary uplink SNR -- i.e. not restricted to symmetric SNRs as in \cite{Narayanan:2010}. 
A similar technique will be utilized in the lattice coding scheme for Two-way Two-relay Channel in Section \ref{sec:procedure}.

The channel model is the same as in \cite{Nam:IEEE}: two users Node $1$ and $3$ communicate with each other through the relay Node $2$. The channel model is expressed as 
\begin{align*}
{\bf Y}_1 &= {\bf X}_2 + {\bf Z}_1 \\
{\bf Y}_2&= {\bf X}_1 + {\bf X}_3 + {\bf Z}_2 \\
{\bf Y}_3 &= {\bf X}_2 + {\bf Z}_3
\end{align*}
where ${\bf Z_i} $ ($i \in \{1,2,3\}$) is an i.i.d. Gaussian noise vector with variance $N_i$: ${\bf Z}_i \sim \mathcal{N} ( {\bf 0}, N_i {\bf  I})$, and the input ${\bf X}_i$ is subject to the transmitting power constraint $P_i$: $\frac{1}{n} E({\bfX}_i^T{\bfX}_i) \leq P_i$. 
Similar definitions of codes and achievability as in Section \ref{sec:model} are assumed. 

We will devise an achievability scheme which uses lattice codes in both the MAC phase and BC phase. For simplicity, to demonstrate the central idea of a lattice-based BC phase which is going to be used in the Two-way Two-relay Channel, we do not use dithers nor  MMSE scaling as in \cite{Erez:2004, Nam:IEEE, Narayanan:2010}\footnote{Dithers and MMSE scaling allows one to go from achieving rates proportional to $\log(SNR)$ to $\log(1+SNR)$. However, we  initially forgo the ``1+'' term for simplicity and so as not to clutter the main idea with additional dithers and MMSE scaling.}. 

We assume that $P_1 = N^2p^2$ and $P_3 = p^2$ where $p \in \mathbb{R}$ 
and $N \in \mathbb{Z}$. 
This assumption will be generalized to arbitrary power constraints in the next section.  
We focus on the symmetric rate for the Two-way Two-relay Channel, i.e. when the coding rates of the two messages are identical. 


{\it Codebook generation:} Consider the messages  $w_a, w_b \in \mathbb{F}_{P_{prime}} = \{0,1,2, \dots, P_{prime} - 1\}$. $P_{prime}$ is a large prime number such that $P_{prime} = [2^{nR_{sym}}]$, where $R_{sym}$ is the symmetric coding rate and $[\;]$ denotes rounding to the nearest prime ($P_{prime} = [2^{nR_{sym}}] \rightarrow \infty$ as $n \rightarrow \infty$ since there are infinitely many primes). The two users Node $1$ and $2$ send the codewords ${\bf X}_1 = Np {\bf t}_a = Np \phi(w_a)$ and ${\bf X}_2 = p {\bf t}_b = p\phi(w_b)$ where $\phi(\cdot)$ is defined in \eqref{eq:phi} in Section \ref{subsec:nested} with the nested lattices $\Lambda\subseteq \Lambda_c$. Notice that their codebooks are scaled versions of the codebook $\mathcal{C}_{\Lambda_c, {\cal V}}$. The symmetric coding rate is then $R_{sym} := \frac{1}{n} \log \frac {V(\Lambda)}{V(\Lambda_c)}$.

In the MAC phase, the relay receives ${\bf Y}_2 = {\bf X}_1 + {\bf X}_3 + {\bf Z}_2$ and decodes $(Np{\bf t}_a + p{\bf t}_b) \mod Np\Lambda$ with arbitrarily low probability of error as $n \rightarrow \infty$ 
with rate constraints 
\begin{align*}
R_{sym} &< \left[ \frac{1}{2} \log \left( \frac{P_1}{N_2}\right) \right]^+\\
R_{sym} &< \left[ \frac{1}{2} \log \left( \frac{P_3}{N_2} \right) \right]^+
\end{align*}
according to Lemma \ref{lem:decodesum}. 

\begin{lemma}
\label{lem:decodesum} 
Let ${\bf X}_a = \alpha \theta {\bf t}_a = \alpha \theta \phi(w_a) \in \{ \alpha \theta \Lambda_{c} \cap \mathcal{V}(\alpha \theta \Lambda) \}$ and ${\bf X}_b = \theta {\bf t}_b =  \theta \phi(w_b) \in \{ \theta \Lambda_{c} \cap \mathcal{V} (\theta \Lambda) \}$ where $w_a, w_b \in \mathbb{F}_{P_{prime}}$, $\alpha \in \mathbb{Z}^+$, $\theta \in \mathbb{R}^+$ and $\phi(\cdot), \Lambda\subseteq \Lambda_c$ are defined as in Section \ref{subsec:nested} with $R := \frac{1}{n} \log \frac{V(\Lambda)}{V(\Lambda_c)}$.
From the received signal ${\bf Y} = {\bf X}_a + {\bf X}_b + {\bf Z}$ where ${\bf Z} \sim \mathcal{N} ({\bf 0}, \sigma^2_z{\bf I})$ one may  decode $(\alpha \theta {\bf t}_a + \theta {\bf t}_b) \mod \alpha \theta \Lambda$ with arbitrary low probability of error as $n \rightarrow \infty$ at rates
\begin{align*}
R&< \left[ \frac{1}{2} \log \frac{\sigma^2(\alpha \theta \Lambda)}{\sigma^2_z} \right]^+\\
R&< \left[ \frac{1}{2} \log \frac{\sigma^2(\theta \Lambda)}{\sigma^2_z}\right]^+.
\end{align*}
\end{lemma}
\begin{proof}
The proof generally follows \cite{Nam:IEEE}, except that for simplicity,  we do not use dithers or MMSE scaling of the received signal in our scheme. The receiver processes the received signal as
\begin{align*}
{\bf Y} \mod \alpha \theta \Lambda &= {\bf X}_a + {\bf X}_b + {\bf Z} \mod \alpha \theta \Lambda\\
&= \alpha \theta {\bf t}_a + \theta {\bf t}_b + {\bf Z} \mod \alpha \theta \Lambda
\end{align*}
To decode $(\alpha \theta {\bf t}_a + \theta {\bf t}_b) \mod \alpha \theta \Lambda$, the effective noise is given by ${\bf Z}$ with variance $\sigma^2_z$ rather than the equivalent noise after MMSE scaling as in  \cite{Nam:IEEE}. All other steps remain identical. The effective signal-to-noise ratios are $SNR_a =  \frac{\sigma^2(\alpha \theta \Lambda_a)}{\sigma^2_z}$ and $SNR_b =  \frac{\sigma^2(\theta \Lambda_b)}{\sigma^2_z}$, resulting in the given rate constraints.
\end{proof}

In the BC phase, if, mimicking the steps of \cite{Narayanan:2010} the relay simply broadcasts the scaled version of $( Np{\bf t}_a + p{\bf t}_b) \mod Np\Lambda$  
\[ \frac{\sqrt{P_2}}{Np} ( \left( Np{\bf t}_a + p{\bf t}_b\right) \mod Np\Lambda )  =  \left(\sqrt{P_2} {\bf t}_a + \frac{\sqrt{P_2}}{N} {\bf t}_b \right) \mod \sqrt{P_2} \Lambda,\] 
we would achieve the rate $R_{sym} < [ \frac{1}{2} \log \frac{P_2}{N_3} ]^+$ for the direction $2 \rightarrow 3$  and the rate  $R_{sym} < [\frac{1}{2} \log \frac{P_2}{N N_1}]^+$ for the $1 \leftarrow 2$ direction. While the rate constraint for the direction $2\rightarrow 3$ is as large as expected,  the rate constraint for the direction $1 \leftarrow 2$ does not fully utilize the power at the relay, i.e. the codeword ${\bf t}_b$ appears to use only the power $P_2/N$ rather than the full power $P_2$. One would thus want to somehow transform the decoded sum $(Np{\bf t}_a + p{\bf t}_b) \mod Np\Lambda$ such that both ${\bf t}_a$ and ${\bf t}_b$ of the transformed signal would be uniformly distributed over $\mathcal{V} ( \sqrt{P_2} \Lambda)$.
Notice that the relay can only operate on $(Np{\bf t}_a + p{\bf t}_b) \mod Np\Lambda$ rather than $Np{\bf t}_a$ and $p{\bf t}_b$   individually. 

{\bf Re-distribution Transform:} To alleviate this problem we propose the following ``Re-distribution Transform'' operation which consists of three steps: 
\begin{enumerate}
\item multiply the decoded signal by $N$ to obtain $ N ( ( Np{\bf t}_a + p{\bf t}_b) \mod Np\Lambda )$, 
\item then take $\mod \Lambda$ to obtain
\[ N ( ( Np{\bf t}_a + p{\bf t}_b) \mod Np\Lambda ) \mod \Lambda = (N^2p{\bf t}_a + Np{\bf t}_b) \mod Np\Lambda \]
 according to the operation rule in \eqref{eq:operation1}, and finally 
 \item re-scale the signal to be of second moment $P_2$ as 
\[ \frac{\sqrt{P_2}}{Np} ( (N^2p{\bf t}_a + Np{\bf t}_b) \mod Np\Lambda )  = (N\sqrt{P_2} {\bf t}_a + \sqrt{P_2} {\bf t}_b ) \mod \sqrt{P_2} \Lambda \]
according to the operation rule in \eqref{eq:operation2}. Notice that $(N\sqrt{P_2} {\bf t}_a + \sqrt{P_2} {\bf t}_b ) \mod \sqrt{P_2} \Lambda$ is uniformly distributed over $\{ \sqrt{P_2} \Lambda_c \cap \mathcal{V} (\sqrt{P_2} \Lambda) \}$  by Lemma \ref{lem:codingrate}.
\end{enumerate}

The three steps of the Re-distribution Transform procedure are illustrated in Figure \ref{fig:illustrate}  for a simple one-dimensional lattice in order to gain some intuition and insight (though we operate in $n$-dimensions in our achievability proof). 
In this simple example,  Node $1$ sends $2pt_a \in \{2p\Lambda_c \cap \mathcal{V}(2p\Lambda) \} = \{-2, -1, 0 ,1,2 \}$ and Node 3 sends $pt_b \in \{p\Lambda_c \cap \mathcal{V}(p\Lambda) \} = \{ -1, -1/2, 0 , 1/2, 1\}$.  The relay decodes \[ (2pt_a + pt_b) \mod 2p\Lambda \in \{p\Lambda_c \cap \mathcal{V}(2p\Lambda) \} = \{-2, -3/2, -1, -1/2, 0 , 1/2, 1, 3/2, 2, 5/2 \} \]
and transforms it into $(4pt_a + 2pt_b) \mod 2p\Lambda \in \{2p\Lambda_c \cap \mathcal{V}(2p\Lambda) \} = \{-2, -1, 0 ,1 ,2\} $ according to the three steps in the Re-distribution Transform.  Notice $|\{2p\Lambda_c \cap \mathcal{V}(2p\Lambda) \} | < | \{p\Lambda_c \cap \mathcal{V}(2p\Lambda) \} |$, and hence the transformed signal is thus easier to forward. 
\begin{figure}
\centering
\includegraphics[width=15cm]{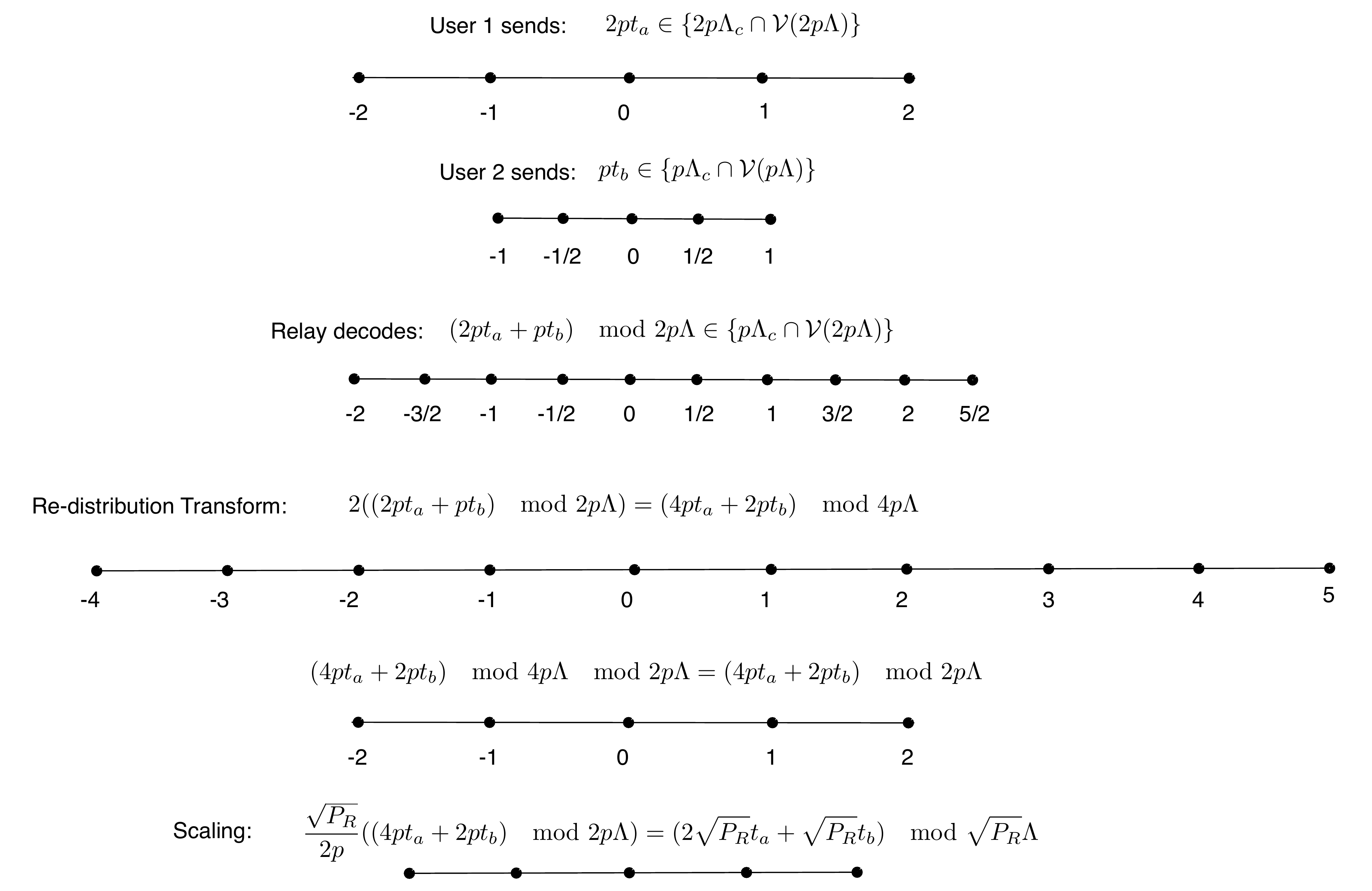}
\caption{Re-distribution Transform illustration for a  one-dimensional lattice.}
\label{fig:illustrate}
\end{figure}

The relay broadcasts ${\bf X}_2  = (N \sqrt{P_2} {\bf t}_a + \sqrt{P_2} {\bf t}_b ) \mod \sqrt{P_2} \Lambda$. Notice that $(N \sqrt{P_2} {\bf t}_a + \sqrt{P_2} {\bf t}_b ) \mod \sqrt{P_2} \Lambda $ is uniformly distributed over $\{ \sqrt{P_2} \Lambda_c \cap \mathcal{V} (\sqrt{P_2} \Lambda) \}$, and so its coding rate is $R_{sym}$. Node $1$ and Node $3$ receive ${\bf Y}_1 = {\bf X}_2 + {\bf Z}_1$ and ${\bf Y}_3 = {\bf X}_2 + {\bf Z}_3$ respectively and, according to Lemma \ref{lem:ptop}, may decode $(N \sqrt{P_2} {\bf t}_a + \sqrt{P_2} {\bf t}_b ) \mod \sqrt{P_2} \Lambda$ 
at rate 
\begin{align*}
R_{sym}  &< \left[ \frac{1}{2} \log \frac{P_2}{N_1} \right]^+ \\
R_{sym}  &< \left [ \frac{1}{2} \log \frac{P_2}{N_3}\right]^+.
\end{align*}

\begin{lemma}
\label{lem:ptop}
 Let ${\bf X} = \theta {\bf t}  = \theta \phi(w)  \in \{ \theta \Lambda_{c} \cap \mathcal{V}(\theta \Lambda) \}$ where $\theta \in \mathbb{R}^+$, $w \in \mathbb{F}_{P_{prime}}$ and $\phi(\cdot), \Lambda \subseteq \Lambda_c$ are defined as in Section \ref{subsec:nested}, with $R = \frac{1}{n} \log \frac{V(\Lambda)}{V(\Lambda_c)}$. 
From the received signal ${\bf Y} = {\bf X} + {\bf Z}$ where ${\bf Z} \sim \mathcal{N} ({\bf 0}, \sigma^2_z{\bf I})$, one may decode $\theta {\bf t}$ with arbitrary low probability of error as $n \rightarrow \infty$ at rate 
\[ R < \frac{1}{2} \log \frac{\sigma^2(\theta \Lambda)}{\sigma^2_z}. \]
\end{lemma}
\begin{proof}
The proof generally follows \cite [Theorem 5] {Erez:2004},  except that we do not use dithers nor MMSE scaling in our scheme.  The receiver processes the received signal as
\begin{align*}
{\bf Y} \mod \Lambda &= {\bf X} + {\bf Z} \mod \Lambda\\
&= {\bf t} + {\bf Z} \mod \Lambda
\end{align*}
To decode ${\bf t}$,  the effective noise is ${\bf Z}$ with variance $\sigma^2_z$ rather than an equivalent noise after MMSE as in \cite [Theorem 5] {Erez:2004}. All other steps are identical. The effective signal-to-noise ratio is thus $SNR =  \frac{\sigma^2(\theta \Lambda)}{\sigma^2_z}$, and we obtain the rate constraints as in the lemma statement.
\end{proof}

Nodes $1$ and $2$ then map the decoded $(NP_R {\bf t}_a + P_R {\bf t}_b ) \mod P_R\Lambda$ to $Nw_a \oplus w_b$ by Lemma \ref{lem:mapping}. With side information $w_a$, Node $1$ may then determine $w_b$; likewise with side information $w_b$, Node $2$ can obtain $Nw_a$ and then determine $w_a$ by Lemma \ref{lem:prime}. 


\section{Two-way Two-relay Channel}
\label{sec:procedure}

We first consider the full-duplex Two-way Two-relay channel where every node transmits and receives at the same time.  For this channel model we first obtain an achievable rate region for special relationships between the power constraints at the four nodes in Theorem \ref{thm:rate}, and the related Lemma \ref{lem:shift}. We then use these results to obtain an achievable rate region for the general (arbitrary powers) Gaussian Two-way Two-relay channel, which we show is to within $\frac{1}{2}\log(3)$ bits/s/Hz per user of the symmetric rate capacity in Theorem \ref{thm:general}. 
\theorem
\label{thm:rate}
For the channel model described in Section \ref{sec:model}, if $P_1 = p^2$, $P_2 = M^2q^2$, $P_3 = N^2p^2$ and $P_4 = q^2$, where $p, q \in \mathbb{R}^+$ and $M, N \in \mathbb{Z}^+$ the following rate region 
\begin{align}
R_a, R_b < \min & \left( \left[ \frac{1}{2} \log \left(\frac{P_1}{N_2}\right) \right]^+, \left[ \frac{1}{2} \log \left(\frac{P_2}{N_3}\right) \right]^+, \left[ \frac{1}{2} \log \left(\frac{P_3}{N_4}\right)\right]^+,  \right. \nonumber \\
& \quad \left. \left[ \frac{1}{2} \log \left(\frac{P_4}{N_3}\right) \right]^+, \left[ \frac{1}{2} \log \left(\frac{P_3}{N_2}\right)\right]^+, \left[ \frac{1}{2} \log \left(\frac{P_2}{N_1}\right) \right]^+ \right) \label{eq:rateconstraints} \end{align}
is achievable using lattice codes. 
\begin{proof} \\
\begin{figure}
\centering
\includegraphics[width=15cm]{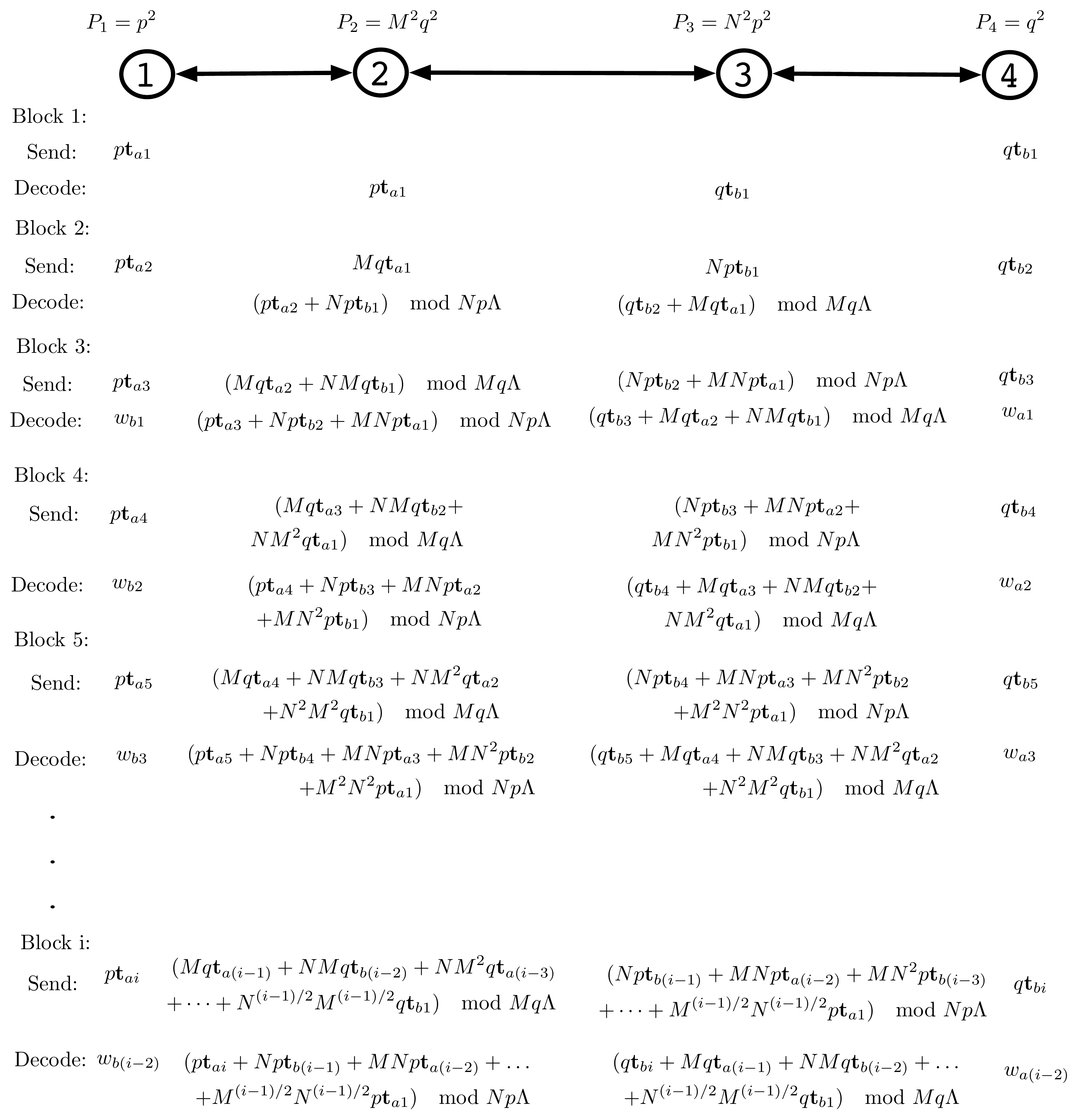}
\caption{Multi-phase Block Markov achievability strategy for Theorem \ref{thm:rate}.}
\label{fig:diagram}
\end{figure}
{\it Codebook generation}: We consider the good nested lattice pair $\Lambda \subseteq \Lambda_c$ with corresponding codebook $\mathcal{C}_{\Lambda_c, {\cal V}} = \{ \Lambda_c \cap \mathcal{V} (\Lambda) \}$, and two messages $w_a, w_b \in \mathbb{F}_{P_{prime}} = \{0,1,2, \dots, P_{prime}-1\}$ in which $P_{prime}$ is a large prime number such that $P_{prime} = [ 2^{nR_{sym}} ]$ ($R_{sym}$ is the coding rate). The codewords associated with the messages $w_a$ and $w_b$ are ${\bf t}_a = \phi(w_a)$ and ${\bf t}_b = \phi(w_b)$, where the mapping $\phi(\cdot)$ from $\mathbb{F}_{P_{prime}}$ to $\mathcal{C}_{\Lambda_c, {\cal V}} \in \mathbb{R}^n$ is defined in \eqref{eq:phi} in  Section \ref{subsec:nested}.

{\it Encoding and decoding steps:} We use a Block Markov Encoding/Decoding scheme where Node $1$ and $4$ transmit a new message $w_{ai}$ and $w_{bi}$, respectively,  at the beginning of block $i$. To satisfy the transmit power constraints,  Node 1 and 4 send the scaled codewords ${\bf X}_{1i} = p{\bf t}_{ai} = p \phi(w_{ai}) \in \{p\Lambda_c \cap \mathcal{V} (p\Lambda) \}$ and ${\bf X}_{4i} = q{\bf t}_{bi} = q\phi(w_{bi}) \in \{q\Lambda_c \cap \mathcal{V} (q\Lambda) \}$ respectively in block $i$. Node $2$ and $3$ send ${\bf X}_{2i}$ and ${\bf X}_{3i}$, and Node $j$ ($j =\{1,2,3,4\}$) receives ${\bf Y}_{ji}$ in block $i$. The procedure of the first few blocks (the initialization steps) are described and then a generalization is made. 
We note that in general the coding rates $R_a$ for $w_a$ and $R_b$ for $w_b$ may be different, as long as $R_{sym} =  \max (R_a, R_b)$, since we may always send dummy messages to make the two coding rates equal. 

Block $1$:  Node $1$ and $4$ send new codewords ${\bf X}_{11} = p{\bf t}_{a1}$ and ${\bf X}_{41} = q{\bf t}_{b1}$ to Node $2$ and $3$ respectively. Node $2$ and $3$ can decode the transmitted codeword with vanishing probability of error if 
\begin{align}
R_{sym} &< \left[ \frac{1}{2} \log \left(\frac{P_1}{N_2}\right) \label{eq:12} \right]^+ \\
R_{sym} &< \left[ \frac{1}{2} \log \left(\frac{P_4}{N_3}\right) \label{eq:43} \right]^+
\end{align}
according to Lemma \ref{lem:ptop}. 

Block $2$:  Node $1$ and $4$ send their respective new codewords ${\bf X}_{12} = p{\bf t}_{a2}$ and ${\bf X}_{42}= q {\bf t}_{b2}$, while Node $2$ and $3$ broadcast ${\bf X}_{22} = Mq{\bf t}_{a1}$ and ${\bf X}_{32} = Np{\bf t}_{b1}$ received in the  last block. Note they are scaled to fully utilize the transmit power $P_2 = M^2q^2$ and $P_3 = N^2p^2$. Node $2$ receives  ${\bf Y}_{22} = {\bf X}_{12} + {\bf X}_{32} + {\bf Z}_{22}$ and decodes $( p{\bf t}_{a2} + Np{\bf t}_{b1} ) \mod Np\Lambda$ with arbitrarily low probability of error if $R$ satisfies \eqref{eq:12} and 
\begin{equation}
R_{sym} <\left[ \frac{1}{2} \log \left(\frac{P_3}{N_2}\right) \right]^+ .\label{eq:32} 
\end{equation}
Similarly, Node $3$ can decode $(q {\bf t}_{b2}+Mq {\bf t}_{a1} ) \mod Mq\Lambda$ subject to \eqref{eq:43} and 
\begin{equation}
R_{sym} < \left[ \frac{1}{2} \log \left(\frac{P_2}{N_3}\right) \right]^+ . \label{eq:23} 
\end{equation}

Block $3$: 
\begin{itemize}
\item {\it Encoding:} Node $1$ and $4$ send new codewords as in the previous blocks.  Node $2$ further processes its decoded codewords combination according to the three steps of the Re-distribution Transform from previous block as 
\begin{align*}
(N  ( ( p{\bf t}_{a2}+ Np{\bf t}_{b1}) \mod Np\Lambda) ) \mod Np \Lambda &= (Np {\bf t}_{a2} + N^2p {\bf t}_{b1} ) \mod N^2p\Lambda \mod Np\Lambda \\
&= (Np{\bf t}_{a2} + N^2p {\bf t}_{b1} ) \mod Np\Lambda
\end{align*}
including scaling to fully utilize the transmit power $P_2 = M^2q^2$ as
\[ \frac{Mq}{Np} (Np {\bf t}_{a2}+ N^2p {\bf t}_{b1} ) \mod Np\Lambda = (Mq {\bf t}_{a2} + NMq {\bf t}_{b1}) \mod Mq\Lambda. \] It then broadcasts ${\bf X}_{23} = (Mq {\bf t}_{a2} + NMq {\bf t}_{b1}) \mod Mq\Lambda$. Notice that since
\[ (Mq {\bf t}_{a2} + NMq {\bf t}_{b1}) \mod Mq\Lambda \in \{ Mq\Lambda_c \cap \mathcal{V} (Mq\Lambda) \}\]
according to Lemma \ref{lem:codingrate}, its coding rate is $R_{sym}$. Similarly, Node 3 broadcasts ${\bf X}_{33} = (Np{\bf t}_{b2}+ MNp{\bf t}_{a1} ) \mod Np\Lambda$ again at  coding rate $R_{sym}$. 
\item {\it Decoding:} At the end of this block,  Node $2$ is able to decode $(p {\bf t}_{a3}+Np{\bf t}_{b2} + MNp {\bf t}_{a1} ) \mod Np\Lambda$ with rate constraints (\ref{eq:12}) and (\ref{eq:32}) according to Lemma \ref{lem:decodesum}, and  Node $3$ decodes $(q {\bf t}_{b3}+Mq{\bf t}_{a2} + NMq {\bf t}_{b1} ) \mod Mq\Lambda$ with constraints (\ref{eq:23}) and (\ref{eq:43}). Node $1$ decodes $(Mq {\bf t}_{a2} + NMq {\bf t}_{b1}) \mod Mq\Lambda$ sent by Node $2$ as in the point-to-point channel with rate constraint
\begin{equation}
R_{sym} < \left[ \frac{1}{2} \log \left(\frac{P_2}{N_1}\right) \right]^+\label{eq:21} 
\end{equation}
according to Lemma \ref{lem:ptop}. From the decoded $(Mq {\bf t}_{a2} + NMq {\bf t}_{b1}) \mod Mq\Lambda$, it obtains $w_{a2} \oplus Nw_{b1}$ (Lemma \ref{lem:mapping}). With its own information $w_{a2}$, Node $1$ can then obtain $N \otimes w_{b1} = w_{a2} \oplus Nw_{b1} \ominus w_{a2}$, which may be mapped to $w_{b1}$ since $P_{prime}$ is a prime number (Lemma \ref{lem:prime}). Notice $P_{prime} = [2^{nR_{sym}}] \rightarrow \infty$ as $ n \rightarrow \infty$, so  $N \ll P_{prime} = [2^{nR}]$ and $\frac{N}{P_{prime}} \notin  \mathbb{Z}$. Similarly, Node $4$ can decode $w_{a1}$ with rate constraint 
\begin{equation}
R_{sym} < \left[ \frac{1}{2} \log \left(\frac{P_2}{N_3}\right) \right]^+ . \label{eq:34} 
\end{equation}
\end{itemize}
Block $4$ and $5$ proceed in a similar manner,  as shown in Figure \ref{fig:diagram}. 

Block i: To generalize, in Block $i$ (assume $i$ is odd)\footnote{For $i$ even we have analogous steps with slightly different indices as may be extrapolated from the difference between Block 4 and 5 in Fig. \ref{fig:diagram}, which result in the same rate constraints.}, 
\begin{itemize}
\item {\it Encoding:} Node $1$ and $4$ send new messages ${\bf X}_{1i} = p{\bf t}_{ai}$ and ${\bf X}_{4i} = q{\bf t}_{bi}$ respectively. 
Node $2$ and $3$ broadcast 
\begin{footnotesize} \begin{align*}
{\bf X}_{2i} &= (Mq {\bf t}_{a(i-1)} + NMq {\bf t}_{b(i-2)} + NM^2q {\bf t}_{a(i-3)} + N^2M^2q {\bf t}_{b(i-4)}+ \dots + N^{(i-1)/2}M^{(i-1)/2} q {\bf t}_{b1} ) \mod Mq\Lambda \\
{\bf X}_{3i} &=  (Np{\bf t}_{b(i-1)} + MNp{\bf t}_{a(i-2)} + MN^2p {\bf t}_{b(i-3)} + M^2N^2p {\bf t}_{a(i-4)}+ \dots + M^{(i-1)/2}N^{(i-1)/2}p{\bf t}_{a1} ) \mod Np\Lambda.
\end{align*}
\end{footnotesize}
\item {\it Decoding:} Node $1$ decodes the codeword from Node $2$ with rate constraint (\ref{eq:21}) (Lemma \ref{lem:ptop}) and maps it to $w_{a(i-1)} \oplus Nw_{b(i-2)} \oplus NMw_{a(i-3)} \oplus N^2Mw_{b(i-4)} \oplus \dots \oplus N^{(i-1)/2}M^{(i-1)/2-1}w_{b1}$ (Lemma \ref{lem:mapping}).  
With its own messages $w_{ai}$ ($\forall i$) and the messages it decoded previously $\{w_{b1}, w_{b2}, \dots, w_{b(i-3)} \}$, Node $1$ can obtain $N \otimes w_{b(i-2)}$ and determine $w_{b(i-2)}$ accordingly (Lemma \ref{lem:prime}). Similarly, Node $4$ can decode $w_{a(i-2)}$ subject to rate constraint (\ref{eq:34}). 

\item {\it Re-distribution Transform:} In this block $i$, Node $2$ also decodes \[ (p{\bf t}_{ai}+ Np {\bf t}_{b(i-1)}+MNp {\bf t}_{a(i-2)} + MN^2p {\bf t}_{b(i-3)} + M^2N^2p {\bf t}_{a(i-4)} + \dots +M^{(i-1)/2}N^{(i-1)/2}p {\bf t}_{a1} ) \mod Np\Lambda\]  from its received signal ${\bf Y}_{2i} = {\bf X}_{1i} + {\bf X}_{3i} + {\bf Z}_{2i}$ subject to rate constraints (\ref{eq:12}) and (\ref{eq:32}) (Lemma \ref{lem:decodesum}). It then uses the Re-distribution Transform to process the codeword combination as 
\begin{align*}
 &( N (p {\bf t}_{ai}+ Np{\bf t}_{b(i-1)}+MNp {\bf t}_{a(i-2)} + \dots +M^{(i-1)/2}N^{(i-1)/2}p{\bf t}_{a1}  \mod Np\Lambda) ) \mod Np\Lambda \\
 = &(Np{\bf t}_{ai}+ N^2p {\bf t}_{b(i-1)}+MN^2p {\bf t}_{a(i-2)} + \dots +M^{(i-1)/2}N^{(i-1)/2+1}p {\bf t}_{a1}  \mod N^2p\Lambda ) \mod Np\Lambda \\
 = &Np{\bf t}_{ai}+ N^2p {\bf t}_{b(i-1)}+MN^2p {\bf t}_{a(i-2)} + \dots +M^{(i-1)/2}N^{(i-1)/2+1}p {\bf t}_{a1}   \mod Np\Lambda
\end{align*}
and scales it to fully utilize  the transmit power:
\begin{align*}
&\frac{Mq}{Np} ( Np{\bf t}_{ai}+ N^2p {\bf t}_{b(i-1)}+MN^2p {\bf t}_{a(i-2)} + \dots +M^{(i-1)/2}N^{(i-1)/2+1}p {\bf t}_{a1}   \mod Np\Lambda) \\
=& Mq{\bf t}_{ai} + NMq {\bf t}_{b(i-1)} + NM^2q{\bf t}_{a(i-2)} + \dots +  N^{(i-1)/2}M^{(i-1)/2+1}q{\bf t}_{a1} \mod Mq\Lambda .
\end{align*}
This signal will be transmitted in the next block $i+1$.  Node 3 performs similar operations, decoding $q{\bf t}_{bi}+Mq {\bf t}_{a(i-1)}+NMq {\bf t}_{b(i-2)}+\dots +N^{(i-1)/2}M^{(i-1)/2}q {\bf t}_{b1} \mod Mq\Lambda$ subject to constraints (\ref{eq:23}) and (\ref{eq:43}), and transforms it into $Np{\bf t}_{bi} + MNp {\bf t}_{a(i-1)} + MN^2p{\bf t}_{b(i-2)} +\dots + M^{(i-1)/2}N^{(i-1)/2+1}p{\bf t}_{b1} \mod Mq\Lambda$, which is transmitted in the next block.
\end{itemize}


Combining all rate constraints, we obtain 
\begin{align*}
 R_{sym} < \min & \left( \left[ \frac{1}{2} \log \left(\frac{P_1}{N_2}\right)\right]^+, \left[ \frac{1}{2} \log \left(\frac{P_2}{N_3}\right)\right]^+, \left[ \frac{1}{2} \log \left(\frac{P_3}{N_4}\right)\right]^+, \right. \\
&\quad \left. \left[ \frac{1}{2} \log \left(\frac{P_4}{N_3}\right)\right]^+, \left[ \frac{1}{2} \log \left(\frac{P_3}{N_2}\right)\right]^+, \left[ \frac{1}{2} \log \left(\frac{P_2}{N_1}\right) \right]^+ \right) \end{align*}
Assuming there are $I$ blocks in total, the final achievable rate is $\frac{I-2}{I}R_{sym}$, which,  as $I \rightarrow \infty$, approaches $R_{sym}$ and we obtain \eqref{eq:rateconstraints}.
\end{proof}

In the above we had assumed power constraints of the form  $P_1 = p^2$, $P_2 = M^2q^2$, $P_3 = N^2p^2$ and $P_4 = q^2$, where $p, q \in \mathbb{R}^+$ and $M, N \in \mathbb{Z}^+$. Analogously, we may permute some of these power constraints to achieve the same region as follows:
\begin{lemma}
\label{lem:shift}
The rates achieved in Theorem \ref{thm:rate} may also be achieved when $P_1 = N^2p^2$, $P_3 = p^2$ and/or  $P_2 = q^2$, $P_4 = M^2 q^2$.

\end{lemma}
\begin{proof}
The proof follows the same lines as that of Theorem \ref{thm:rate}, and consists of the steps outlined in 
 in Figure \ref{fig:diagram2} in Appendix \ref{app:shift}. In particular, since the nodes have different power constraints the scaling of the codewords is different. However, as in the previous Theorem, 
 we only need ${\bf X}_1$ and ${\bf X}_3$ to be aligned (nested codebooks), and ${\bf X}_2$ and ${\bf X}_4$ to be aligned. 
 As in Theorem \ref{thm:rate}, the relay nodes again decode the sum of codewords, perform the Re-distribution Transform,  with a new power scaling to fully utilize their transmit power, and broadcast the re-distributed sum of codewords. 
\end{proof}

Theorem \ref{thm:rate} and Lemma \ref{lem:shift} both hold for powers for which
 $P_1/P_3$ and/or $P_2/P_4$ are either the squares of integers or the reciprocal of the squares of integers. However, these scenarios do not cover general power constraints with arbitrary ratios. We next present an achievable rate region for arbitrary powers, obtained by appropriately clipping the power of the nodes such that the new, lower powers are indeed either the squares or the reciprocals of the squares of integers. We then show that this clipping of the power at the nodes does not result in more than a $\frac{1}{2}\log(3)$ bits/s/Hz loss in the symmetric rate. 
 
 \theorem
\label{thm:general}
For the Two-way Two-relay Channel with arbitrary transmit power constraints, any rates satisfying
\begin{align}
R_a, R_b < R_{achievable} =  & \max_{P_i' \leq P_i, \frac{P_1' }{P_3' } = N^2 \text{ or } \frac{1}{N^2}, \frac{P_2' }{P_4' } = M^2 \text{ or } \frac{1}{M^2}  } \min \left( \left[ \frac{1}{2} \log \left(\frac{P_1'}{N_2}\right) \right]^+, \left[ \frac{1}{2} \log \left(\frac{P_2'}{N_3}\right)\right]^+, \right.  \nonumber \\
& \left.  \quad \quad \left[ \frac{1}{2} \log \left(\frac{P_3'}{N_4}\right)\right]^+,\left[\frac{1}{2} \log \left(\frac{P_4'}{N_3}\right) \right]^+, \left[\frac{1}{2} \log \left(\frac{P_3'}{N_2}\right)\right], \left[ \frac{1}{2} \log \left(\frac{P_2'}{N_1}\right) \right]^+ \right) \label{eq:achieve}
\end{align}
for some $N, M \in \mathbb{Z}^+$  and $ i \in \{1,2,3,4\}$, are achievable. This rate region is within $\frac{1}{2} \log 3$ bit/Hz/s per user from the symmetric rate capacity.
\begin{proof}
Since $P_1/P_3$ and/or $P_2/P_4$ are in general neither the squares or the reciprocals of the squares of integers, we cannot directly apply Theorem \ref{thm:rate}. Instead, we first truncate the transmit powers $P_i$ to $P_i'$ such that  $P_i'$ satisfy  either the constraints of Theorem \ref{thm:rate} or Lemma \ref{lem:shift}. The achievable rate region then follows immediately for the reduced power constraints $P_i'$. For example, if $P_1 = 1$ and $P_3 = 3.6$, we may choose $P_1' = 0.9$ and $P_3' = 3.6$ , so that $P_3'/P_1' = 2^2$. Optimizing over the truncated or clipped powers yields the achievable rate region stated in Theorem \ref{thm:general}.

An outer bound to the symmetric capacity of the AWGN Two-way Two-relay Channel is given by the minimum of the all the point-to-point links, i.e.
\begin{equation}
R_a, R_b < R_{outer} = \min \left(C\left(\frac{P_1}{N_2}\right), C\left(\frac{P_2}{N_3}\right), C\left(\frac{P_3}{N_4}\right), C\left(\frac{P_4}{N_3}\right), C\left(\frac{P_3}{N_2}\right), C\left(\frac{P_2}{N_1}\right) \right).  \label{eq:outer}
\end{equation}
To evaluate the gap between (\ref{eq:achieve}) and (\ref{eq:outer}): 
\begin{align}
&R_{achievable} + \frac{1}{2} \log 3 \\
\overset{(a)}{=}&\max_{P_1' \leq P_1, P_3' \leq P_3, \frac{P_1' }{P_3' } = N^2 \text{ or } \frac{1}{N^2}} \min \left( \left[ \frac{1}{2} \log \left(\frac{P_{1}'}{N_2}\right) \right]^+, \left[ \frac{1}{2} \log \left(\frac{P_3'}{N_4}\right)\right]^+, \left[ \frac{1}{2} \log \left(\frac{P_3'}{N_2}\right) \right]^+ \right) + \frac{1}{2} \log 3 \label{eq:1}\\
=& \max_{P_1' \leq P_1, P_3' \leq P_3, \frac{P_1' }{P_3' } = N^2 \text{ or } \frac{1}{N^2}} \max \left( \min \left( \frac{1}{2} \log \left(\frac{3P_{1}'}{N_2}\right), \frac{1}{2} \log \left(\frac{3P_3'}{N_4}\right), \frac{1}{2} \log \left(\frac{3P_3'}{N_2}\right)\right), \frac{1}{2} \log 3 \right) \label{eq:2} \\
\overset{(b)}{\geq}&  \max \left( \min \left( \frac{1}{2} \log \left(\frac{3P_{1}^{\star} }{N_2}\right), \frac{1}{2} \log \left(\frac{3P_3^{\star}}{N_4}\right), \frac{1}{2} \log \left(\frac{3P_3^{\star}}{N_2}\right)\right), \frac{1}{2} \log 3\right) \label{eq:4}\\
\overset{(c)}{\geq}& \min \left( \frac{1}{2} \log \left(1+ \frac{P_1}{N_2} \right), \frac{1}{2} \log \left(1+ \frac{P_3}{N_4} \right), \frac{1}{2} \log \left(1+ \frac{P_3}{N_2} \right) \right)  \label{eq:5}\\
\geq& R_{outer}.
\end{align}
The first equality (a) follows from an assumption that WLOG, one of $\left[ \frac{1}{2} \log \left(\frac{P_{1}'}{N_2}\right)\right]^+,$ $ \left[ \frac{1}{2} \log \left(\frac{P_3'}{N_4}\right)\right]^+,$ or $ \left[ \frac{1}{2} \log \left(\frac{P_3'}{N_2}\right) \right]^+$ is the tightest constraint.  
The first inequality (b) follows since the rates achieved by the optimized powers must be larger than those achieved by one particular strategy that meets the constraints -- in this case the strategy that yields the $P_i^{\star}$ which we construct  in Appendix \ref{app:gap}. Inequality (c) follows from the fact that $2P_i^{\star} \geq P_i$, as also shown in Appendix \ref{app:gap}. 
Finally, we bound $ \frac{1}{2} \log \left(\frac{3P_{1}^{\star}}{N_2}\right)$ with 
$\frac{1}{2} \log \left(1+ \frac{P_1}{N_2} \right)$ as follows: If $\frac{P_1}{N_2} \geq 2$, it follows that $\frac{P_1^{\star}}{N_2} \geq 1$, and hence  $\frac{1}{2} \log \left(\frac{3P_{1}^{\star}}{N_2} \right) \geq \frac{1}{2} \log \left(1+ \frac{P_1}{N_2} \right) $. Otherwise,  $\frac{1}{2} \log \left(1+ \frac{P_1}{N_2} \right) < \frac{1}{2} \log 3$. Similarly, we may bound $ \frac{1}{2} \log \left(\frac{3P_{3}^{\star}}{N_4}\right)$ and $ \frac{1}{2} \log \left(\frac{3P_{3}^{\star}}{N_2}\right)$ with $\frac{1}{2} \log \left(1+ \frac{P_3}{N_4} \right)$ and $\frac{1}{2} \log \left(1+ \frac{P_3}{N_2} \right)$ respectively. \\
\end{proof}

\begin{remark}
We achieve a constant gap to capacity for the symmetric rate of $\frac{1}{2} \log (3)$ bit/s/Hz. The only other scheme we are aware of that has been shown to achieve a constant gap is noisy network coding
 \cite{lim2011noisy} which achieves a larger gap of $1.26$ bits/s/Hz to the capacity. The improvement in rates of the proposed scheme may be attributed to the removal of the noise at intermediate relays (it is a lattice based Decode and Forward scheme) without the sum-rate constraints that would be needed in i.i.d. random coding based Decode-and-Forward schemes. 
\end{remark}

\section{Extensions to half-duplex channels and more than two relays}
\label{sec:extra}

We now extend our results to half-duplex channels and to two-way relay channels with more than two relays. 

\subsection{Half-duplex}

The proposed lattice coding scheme may be generalized to channels with half-duplex nodes, i.e. in which a node may either transmit or receive at a given time but not both.  The scheme is illustrated in Figure \ref{fig:halfduplex}, where we see that the proof generally mimics the full-duplex case (Theorem \ref{thm:rate}, Lemma \ref{lem:shift} and Theorem \ref{thm:general}), but that each  phase (or block) in the full-duplex case is divided into two phases / blocks in the half-duplex case, as nodes may not transmit and receive at the same time.  Thus, one may achieve half the rates as in the full duplex case, i.e. 
\[ R_a, R_b < R_{half-duplex} =  \frac{1}{2} R_{achievable}, \]
where $R_{achievable}$ is expressed in Theorem \ref{thm:general}.

\subsection{More than two relays}
This lattice coding scheme may also be generalized to more than two relays. For example, for the Two-way Three-relay Channel with five nodes: $1 \leftrightarrow 2 \leftrightarrow 3 \leftrightarrow 4 \leftrightarrow 5$, we may apply the same strategy by aligning the codewords between $1$ and $3$, $3$ and $5$, and $2$ and $4$ (i.e. nest the corresponding codebooks). To align or nest the codebooks, one may truncate the transmit powers so that $\frac{P_1'}{P_3'}$, $\frac{P_3'}{P_5'}$ and $\frac{P_2'}{P_4'}$ are either  squares, or the reciprocals of squares of integers.  By extending our Block Markov strategy, the final achievable rate region would be:
\begin{align*}
R_a, R_b <   \max_{\tiny \begin{array}{c}P_i' \leq P_i, \\ \frac{P_1' }{P_3' } = N^2 \text{ or } \frac{1}{N^2}, \\ \frac{P_2' }{P_4' } = M^2 \text{ or } \frac{1}{M^2}, \\  \frac{P_3' }{P_5' } = K^2 \text{ or } \frac{1}{K^2} \end{array}} \min_{\tiny \begin{array}{c} k=\{1,2,3,4\}, \\  j=\{2,3,4,5\}\end{array}} & \left( \left[ \frac{1}{2} \log \left( \frac{P_{k}'}{N_{k+1}}\right) \right]^+, 
\left[ \frac{1}{2} \log \left( \frac{P_{j}'}{N_{j-1}} \right) \right]^+ \right). 
\end{align*}
where $M, N, K \in \mathbb{Z}^+$ and $i= \{1,2,3,4,5\}$.
\begin{figure}
\centering
\includegraphics[width=10cm]{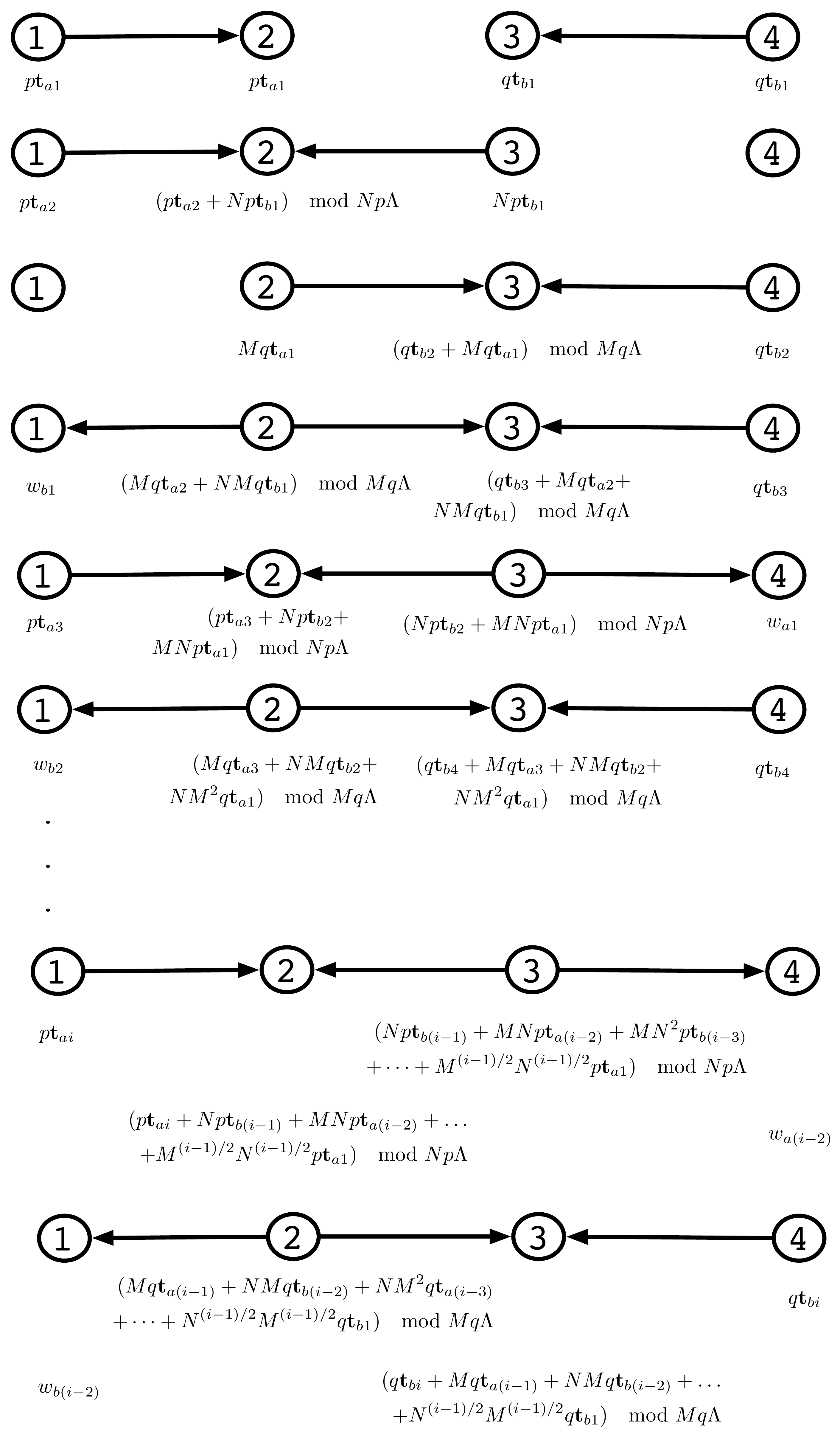}
\caption{Half-duplex case.}
\label{fig:halfduplex}
\end{figure}


\section{Conclusion}
\label{sec:conclusion}

We proposed a lattice coding scheme for the AWGN Two-way Two-relay Channel ($1\leftrightarrow 2 \leftrightarrow 3 \leftrightarrow 4$) which achieves within $\frac{1}{2} \log 3$ bit/Hz/s from the symmetric rate capacity. This scheme may be generalized to half-duplex nodes, and two-way channels with more than two relays.  In our scheme,  each relay decodes a sum of codewords as all transmitted signals are properly chosen lattice codewords, performs the ``Re-distribution Transform" which maps the decoded lattice point to another so as to fully utilize its transmit power, and broadcasts this transformed, scaled, lattice codeword. All decoders are lattice decoders and only a single nested lattice codebook pair is needed in our scheme. 
Note that this work results in a symmetric rate region which means the rates of both directions have to satisfy the constraints of all the links. One interesting open question is whether one can derive an achievability scheme which would result in an asymmetric rate region (where rates of each direction are only  constrained by links for that direction and the two directions would thus not interfere), and whether this would result in better rates. 

\appendix
\subsection{Multi-phase Block Markov achievability strategy for  Lemma \ref{lem:shift}}
\label{app:shift}
The achievable rate region for Lemma \ref{lem:shift} is the same as that for Theorem \ref{thm:rate}; the achievability strategy is essentially the same, with slight variations on the re-scaling to meet the different power constraints. The details of who transmits and decodes what in each phase is outlined in Figure \ref{fig:diagram2}.
\begin{figure}
\centering
\includegraphics[width=14cm]{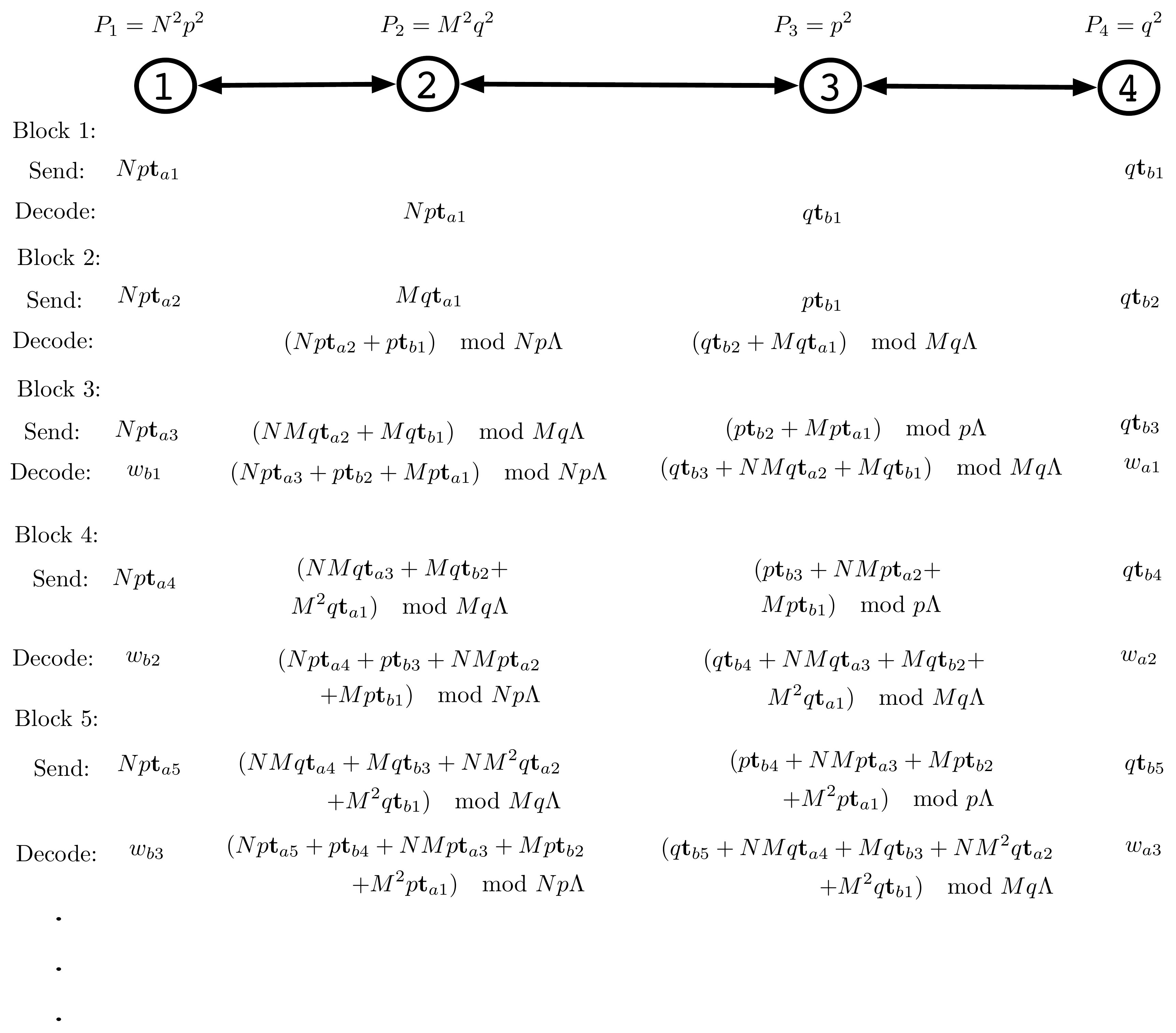}
\caption{Multi-phase Block Markov achievability strategy for  Lemma \ref{lem:shift}.} 
\vspace{-1cm}
\label{fig:diagram2}
\end{figure}

\subsection{A particular choice of $P_i^{\star}$ and proof that $2P_{i}^{\star} \geq P_{i}$ ($i \in \{1,3\}$) in Theorem \ref{thm:general} }
\label{app:gap}
WLOG, we assume $P_3 \geq P_1$. Then, $m^2 \leq \frac{P_3}{P_1} \leq (m+1)^2$ for some integer $m \in \mathbb{Z}^+$. Consider 
the following strategy for choosing $P_i^{\star}$ such that $\frac{P_3^{\star}}{P_1^{\star}}$ is a non-zero integer squared or the inverse 
of an integer squared: If $m^2 \leq \frac{P_3}{P_1}\leq m(m+1)$, we choose $ P_3^{\star}= m^2P_1$ and $ P_1^{\star}= P_1$. Then $ \frac{P_3^{\star}}{P_3} = \frac{m^2P_1}{P_3} \geq \frac{m^2P_1}{m(m+1)P_1} \geq \frac{1}{2}$. Thus, $2P_3^{\star} \geq P_3$ and 
$P_1^{\star} = P_1$. Otherwise if $m(m+1) \leq \frac{P_3}{P_1}\leq (m+1)^2$, we choose $ P_1^{\star} = \frac{1}{(m+1)^2} P^3$ and $ P_3^{\star} = P_3$. Then $\frac{P_1^{\star}}{P_1} = \frac{P_3}{(m+1)^2P_1} \geq \frac{m(m+1)P_1}{(m+1)^2 P_1} \geq \frac{1}{2}$. Thus $2P_1^{\star} \geq P_1$ and $P_3^{\star} = P_3$. In general, this strategy ensures that $2P_i^{\star} \geq P_i$. 

\bibliographystyle{IEEEtran}
\bibliography{refs}
\end{document}